%% file: paper-arxiv.tex
\documentclass[a4paper,11pt]{article}
\usepackage[margin=1in]{geometry}
\usepackage{microtype}
\usepackage{hyperref}
\usepackage[utf8]{inputenc}
\usepackage{thm-restate}
\usepackage{amsthm}
\usepackage{amsmath}
\usepackage{amsfonts}
\usepackage{amssymb}
\usepackage{bbm}
\usepackage{url}
\usepackage{array}
\usepackage{graphicx}
\usepackage{latexsym}

\newcommand{\mysize}{\normalsize}

\newcommand{\bfP}{\mathbf{P}}
\newcommand{\neigh}[1]{\emph{neigh}(#1)}
\newcommand{\E}{\mathbb{E}}
\newcommand{\prob}{\text{Pr}}

\newcommand{\poly}{\operatorname{poly}}

\usepackage{algorithm}
\usepackage[noend]{algpseudocode}
\newcommand{\step}{\textit{step}}
\newcommand{\probe}{\textit{probe}}
\newcommand{\tvd}[2]{\| {#1 - #2} \|_\textsc{TV}}
\newcommand{\sumest}{SumApprox}
\newcommand{\taupiest}{MassApprox}
\newcommand{\taunest}{FullMassApprox}
\newcommand{\kerr}{k_{\epsilon,\delta}}

\newtheorem{theorem}{Theorem}
\newtheorem{lemma}[theorem]{Lemma}

\newtheorem{corollary}{Corollary}

\title{\textbf{On approximating the stationary distribution of time-reversible Markov chains}\footnote{This is the full version of the STACS 2018 paper~\cite{Bressan&2018}. Last update: 2017/12/30.}}

\author{
{\mysize Marco Bressan}\\
{\mysize Dip.\ Informatica}\\
{\mysize Sapienza Universit\`a di Roma}\\
{\mysize Roma, Italy}\\
{\mysize bressan@di.uniroma1.it}
\and
{\mysize Enoch Peserico}\\
{\mysize Dip.\ Ing.\ Informazione}\\
{\mysize Universit\`a di Padova}\\
{\mysize Padova, Italy}\\
{\mysize enoch@dei.unipd.it}
\and
{\mysize Luca Pretto}\\
{\mysize Dip.\ Ing.\ Informazione}\\
{\mysize Universit\`a di Padova}\\
{\mysize Padova, Italy}\\
{\mysize pretto@dei.unipd.it}
}

\date{}

\begin{document}
\maketitle

\begin{abstract}
Approximating the stationary probability of a state in a Markov chain through Markov chain Monte Carlo techniques is, in general, inefficient. 
Standard random walk approaches require $\tilde{O}(\tau/\pi(v))$ operations to approximate the probability $\pi(v)$ of a state $v$ in a chain with mixing time $\tau$, and even the best available techniques still have complexity $\tilde{O}(\tau^{1.5}/\pi(v)^{0.5})$; and since these complexities depend inversely on $\pi(v)$, they can grow beyond any bound in the size of the chain or in its mixing time.
In this paper we show that, for time-reversible Markov chains, there exists a simple randomized approximation algorithm that breaks this ``small-$\pi(v)$ barrier''.
\\[.01\textheight]
\end{abstract}

\section{Introduction}
\label{sec:intro}
We investigate the problem of approximating efficiently a \textit{single} entry of the stationary distribution of an ergodic Markov chain.
This problem has two main motivations.
First, with the advent of massive-scale data, even complexities linear in the size of the input are often excessive~\cite{Rubinfeld&2011}; therefore computing explicitly the entire stationary distribution, e.g.\ via the power method~\cite{golub2012matrix}, can be simply infeasible.
As an alternative one can then resort to approximating only individual entries of the vector, in exchange for a much lower computational complexity~\cite{Lee&2014,Shyamkumar&}.
In fact, if such a complexity is low enough one could efficiently ``sketch'' the whole vector by quickly getting a fair estimate of its entries.
Second, in many practical cases one is really interested in just a few entries at a time.
A classic example is that of network centralities, many of which are stationary distributions of an ergodic Markov chain~\cite{Bonacich&2001}.
Indeed, the problem of approximating the Personalized PageRank score of a few nodes in a graph has been repeatedly addressed in the past~\cite{Borgs&2012,Borgs&2014,Lofgren&2014b,Lofgren&2015}.

In this paper we seek for efficient algorithms for approximating the stationary probability $\pi(v)$ of some target state $v$ in the state space of a discrete-time ergodic Markov chain.
Besides the motivations above, the problem arises in estimating heat kernels and graph diffusions, testing the conductance of graphs and chains, developing local algorithms, and has applications in machine learning; see~\cite{Lofgren&2015b,Lee&2013} for a thorough discussion.
We adopt a simple model where with a single operation one can either (i) simulate one step of the chain or (ii) retrieve the transition probability between a pair of states.
Although recent research has provided encouraging results, existing algorithms suffer from a crucial bottleneck: to guarantee a small \emph{relative} error in the approximation of $\pi(v)$, they incur a cost that grows with $1/\pi(v)$ itself (basically because estimating $\pi(v)$ via repeated sampling requires $1/\pi(v)$ samples).
This is a crucial issue since in general there is no lower bound on $\pi(v)$; even worse, if the state space has $n$ states, then most states have mass $\pi(v) = O(\frac{1}{n})$, and one can easily design chains where they have mass exponentially small in $n$.
In general, then, the cost of existing algorithms can blow up far beyond $O(n)$ for almost all input states $v$.
It is thus natural to ask if the dependence of the complexity on $\pi(v)$ is unavoidable.
Unfortunately, one can easily show that $\Omega(\tau/\pi(v))$ operations can be necessary to estimate $\pi(v)$ within any constant multiplicative factor if one makes no assumption on the chain (see Appendix~\ref{sub:general_impossible}).
To drop below this complexity barrier one must then necessarily look at special classes of Markov chains.

We present an algorithm that breaks this ``small-$\pi(v)$ barrier'' for \emph{time-reversible} Markov chains. Time-reversible chains are a well-known subclass of Markov chains which lie at the heart of the celebrated Metropolis-Hastings algorithm~\cite{Hastings1970} and are equivalent to random walks on weighted undirected graphs~\cite{Levin&2006}.
Formally, given any $\epsilon,\delta>0$ and any state $v$ in a time-reversible chain, our algorithm with probability $1-\delta$ returns a multiplicative $(1\pm\epsilon)$-approximation of $\pi(v)$ by using $\tilde{O}(\tau \|\pi\|^{-1} )$ operations, where $\tau$ is the mixing time of the chain, $\|\cdot\|$ is the Euclidean norm and $\tilde{O}(\cdot)$ hides polynomials in $\epsilon^{-1}$, $\ln{\!(\delta^{-1})}$, $\ln{\!(\|\pi\|^{-1})}$.
The complexity is independent of $\pi(v)$, and for all but a vanishing fraction of states in the chain improves by factors at least $\sqrt{n}$ or $\sqrt{\tau}$ over previous algorithms.
The heart of our algorithm is a randomized scheme for approximating the sum of a nonnegative vector by sampling its entries with probability proportional to their values.
This scheme requires $\tilde{O}(\|\pi\|^{-1})$ samples if $\pi$ is the distribution over the vector entries, which generalizes the $O(\sqrt{n})$ algorithm of~\cite{Motwani&2007} and is provably optimal.
We prove that our algorithm for estimating $\pi(v)$ is essentially optimal as a function of $\tau$, $n$ and $\|\pi\|$; in fact one cannot do better even under a stronger computational model where all transition probabilities to/from all visited states are known.
Finally, we show the number of distinct states visited by our algorithm may be further reduced, provided such a number satisfies some concentration hypotheses.
This is useful if visiting a new state is expensive (e.g.\ if states are users in a social network).
All our algorithms are simple to implement, require no tuning, and an experimental evaluation shows that in practice they are faster than existing alternatives already for medium-sized chains (see Appendix~\ref{sub:exp}).

The rest of the paper is organized as follows.
Subsection~\ref{sub:notation} pins down definitions and notation; Subsection~\ref{sub:problem} formalizes the problem; Subsection~\ref{sub:rel} discusses related work; Subsection~\ref{sub:results} summarizes our results.
Section~\ref{sec:sumapprox} presents our vector sum approximation algorithm.
Section~\ref{sec:massapprox} presents our approximation algorithm for $\pi(v)$.
Section~\ref{sec:lb} provides the proofs of optimality.
All missing details are found in the Appendix.

\subsection{Preliminaries}
\label{sub:notation}
A discrete-time, finite-state Markov chain is a sequence of random variables $X_0,X_1,\ldots$ taking value over a set of states $V = \{1, \ldots, n\}$, such that for all $i\ge 1$ and all $u_0,\ldots,u_i \in V$ with $\prob(X_0 = u_0, \ldots, X_{i-1} = u_{i-1}) > 0$ we have $\prob(X_i = u_i | X_0 = u_0, \ldots, X_{i-1} = u_{i-1}) = \prob(X_i = u_i | X_{i-1} = u_{i-1})$.
Denote by $\bfP = [p_{uu'}]$ the transition matrix of the chain, so that $p_{uu'} = \prob(X_{i} = u' | X_{i-1} = u)$.
We assume the chain is \emph{ergodic}, and thus has a limit distribution that is independent from the distribution of $X_0$; the limit distribution then coincides with the \emph{stationary distribution} $\pi$.
Thus $\pi$ is the unique distribution vector such that for any distribution vector $\pi_0$:
\begin{align}
\pi = \pi \bfP = \lim_{t \to \infty} \pi_0\bfP^t
\end{align}
We denote by $\pi(u)$ the stationary probability, or \emph{mass}, of $u$, and we always denote by $v$ the target state whose mass is to be estimated. For any $V' \subseteq V$ we let $\pi(V')$ denote $\sum_{u \in V'} \pi(u)$.
We also assume the chain is \emph{time-reversible}, i.e.\ that for any pair of states $u$ and $u'$ we have:
\begin{align}
\label{eqn:detbalance}
\pi(u) p_{uu'} = \pi(u') p_{u'u}
\end{align}
We denote by $\tau$ the standard $\frac{1}{4}$-\emph{mixing time} of the chain.
In words, $\tau$ is the smallest integer such that after $\tau$ steps the total variation distance between $\pi$ and the distribution of $X_{\tau}$ is bounded by $\frac{1}{4}$, irrespective of the initial distribution.
Formally,
$\tau := \min\{t : d(t) \le \frac{1}{4}\}$,
where
\begin{align}
d(t) := \max_{\pi_0} \tvd{\pi_0 {\bfP}^t}{\pi} = \max_{\pi_0} \frac{1}{2} \|\pi_0 {\bfP}^t - {\pi}\|_1
\end{align}
After $\tau$ steps, the distribution of $X_t$ converges to $\pi$ exponentially fast; that is, if $t = \eta \tau$ with $\eta \ge 1$, then $\tvd{\pi_0 {\bfP}^t}{\pi} \le 2^{-\eta}$.
In the rest of the paper, $\|\cdot\|$ always denotes the $\ell^2$ norm.
One may refer to~\cite{Levin&2006} for a detailed explanation of the notions recalled here.

Unless necessary, we drop multiplicative factors depending only on $\epsilon,\delta$ (see below) from the asymptotic complexity notation.
Furthermore, we use the tilde notation to hide polylogarithmic factors, i.e.\ we denote $O(f \cdot \poly(\log(f)))$ by $\tilde{O}(f)$.

\subsection{Problem formulation}
\label{sub:problem}
Consider now a discrete-time, finite-state, time-reversible, ergodic Markov chain on $n$ states.
The chain is initially unknown and can be accessed via two operations (also called queries):
\\[3pt]
\hspace*{1em}\step(): accepts in input a state $u$, and returns state $u'$ with probability $p_{uu'}$
\\[3pt]
\hspace*{1em}\probe(): accepts in input a pair of states $u,u'$, and returns $p_{uu'}$
\\[3pt]
These queries are the \textit{de facto} model of previous work.
\step() is used in~\cite{Borgs&2012,Lee&2013,Borgs&2014,Lofgren&2014b,Lofgren&2015,Lofgren&2015b} to simulate the walk, assuming each step costs $O(1)$.
\probe() is used in~\cite{Lofgren&2014b,Lofgren&2015,Lofgren&2015b} to access the elements of the transition matrix, assuming again one access costs $O(1)$.
Here, too, we assume \step() and \probe() as well as all standard operations (arithmetics, memory access, \ldots) cost $O(1)$.
This includes set insertion and set membership testing; in case their complexity is $\omega(1)$, our bounds can be adapted correspondingly.
The problem can now be formalized as follows.
The algorithm is given in input a triple $(v,\epsilon,\delta)$ where $v$ is a state in the state space of the chain and $\epsilon,\delta$ are two reals in $(0,1)$.
It must output a value $\hat{\pi}(v)$ such that, with probability $1-\delta$, it holds $(1-\epsilon)\pi(v) \le \hat{\pi}(v) \le (1+\epsilon)\pi(v)$.
The complexity of the algorithm is counted by the total number of operations it performs.
Obviously we seek for an algorithm of minimal complexity.

A final remark.
We say state $u$ has been \emph{visited} if $u = v$ or if $u$ has been returned by a \step() call.
In line with previous work, we adopt the following ``locality'' constraint: the algorithm can invoke \probe() and \step() only on visited states.

\subsection{Related work}
\label{sub:rel}
Two recent works address precisely the problem of estimating $\pi(v)$ in Markov chains.
The key differences with our paper are that they consider general (i.e.\ not necessarily time-reversible) chains, and that we aim at a small \emph{relative} error for \emph{any} $\pi(v)$ and not only for large $\pi(v)$.
\begin{itemize}
\item \cite{Lee&2013} gives a local approximation algorithm based on estimating return times via truncated random walks.
Given any $\Delta > 0$, if $\pi(v) \ge \Delta$ the algorithm with probability $1-\delta$ outputs a multiplicative $\epsilon Z(v)$-approximation of $\pi(v)$, where $Z(v)$ is a ``local mixing time'' that depends on the structure of the chain.
The cost is $\tilde{O}(\ln{\!(1/\delta)} / \epsilon^3 \Delta)$ \step() calls.
If one wants a multiplicative $(1\pm\epsilon)$-approximation of $\pi(v)$ for a generic $v$, the cost becomes $\tilde{O}(\tau / \pi(v))$ \step() calls since one must wait for the walks to hit $v$ after having mixed.

\item \cite{Lofgren&2015b} gives an algorithm to approximate $\ell$-step transition probabilities based on coupling a local exploration of the transition matrix $\bfP$ with simulated random walks.
Given any $\Delta>0$, if the probability to be estimated is $\ge \Delta$ then with probability $1-\delta$ the algorithm gives a multiplicative $(1\pm\epsilon)$-approximation of it at an expected cost of $\tilde{O}(\ell^{1.5} \sqrt{d\, \ln{\!(1/\delta)}} \,/\, \epsilon \Delta^{0.5})$ calls to both \step() and \probe(), for a uniform random choice of $v$ in the chain, where $d$ is the density of $\bfP$.
To estimate $\pi(v)$ for a generic $v$ one must set $\ell = \tau$ and $\Delta = \pi(v)$, and since if the chain is irreducible then $d=\Omega(1)$, the bound stays at $\tilde{O}(\tau^{1.5} / \pi(v)^{0.5})$.
This does not contradict our lower bound of Appendix~\ref{sub:general_impossible}, since their model allows for probing transition probabilities even between unvisited states.
\end{itemize}

Similar results are known for specific Markov chains, and in particular for PageRank (note that in PageRank $\tau = O(1)$).
\cite{Borgs&2012,Borgs&2014} give an algorithm for approximating the PageRank $\pi(v)$ of the nodes $v$ having $\pi(v) \ge \Delta$, at the cost of $\tilde{O}(1/\Delta)$ \step() calls; again, if one desires a multiplicative $(1\pm\epsilon)$-approximation of $\pi(v)$, the cost becomes $\tilde{O}(1/\pi(v))$.
\cite{Lofgren&2014b} gives an algorithm, with techniques similar to~\cite{Lofgren&2015b}, for estimating the Personalized PageRank $\pi(v)$ of a node $v$; if one aims at a multiplicative $(1\pm\epsilon)$-approximation of $\pi(v)$, the algorithm makes $\tilde{O}(d^{\,0.5} / \pi(v)^{0.5})$ \step() and \probe() calls where $d$ is the average degree of the graph.
Similar bounds can be found in~\cite{Lofgren&2015} for Personalized PageRank on undirected graphs.

Summarizing, existing algorithms require either $\tilde{O}(\tau / \pi(v))$ or $\tilde{O}(\tau^{1.5} / \pi(v)^{0.5})$ \step() and \probe() calls to ensure a $(1\pm\epsilon)$-approximation of $\pi(v)$ for a generic state $v$.
Note that the complexity and approximation guarantees of these algorithms depend on knowledge of $\tau$; our algorithms are no exception, and we prove our bounds as a function of $\tau$.

Finally, for the problem of estimating the sum of a nonnegative $n$-entry vector $\mathbf{x}$ by sampling its entries $x_i$ with probability $\pi_i = x_i / \sum_{i}x_i$, the only algorithm existing to date is that of~\cite{Motwani&2007}.
That algorithm takes $O(\sqrt{n})$ samples independently of $\pi$, while ours needs $O(\sqrt{n})$ samples only in the worst case, i.e.\ if $\pi$ is (essentially) the uniform distribution.

\subsection{Our results}
\label{sub:results}
Our first contribution is \sumest, a randomized algorithm for estimating the sum $\gamma$ of a nonnegative vector $\mathbf{x}$, assuming one can sample its entries according to the probability distribution $\pi = \mathbf{x} / \gamma$.
Formally, we prove:
\begin{restatable}{theorem}{sumapprox}
\label{thm:sumapprox}
Given any $\delta,\epsilon \in (0,1)$, \sumest($\epsilon,\delta$) with probability at least $1-\delta$ returns a multiplicative $(1\pm \epsilon)$-approximation of $\gamma$ by taking $O\big(\|\pi\|^{-1} \epsilon^{-3}(\ln{\frac{1}{\delta}})^{3/2}\big)$ samples.
\end{restatable}
\noindent \sumest\ is extremely simple, yet it improves on the state-of-the-art $O(\sqrt{n})$ algorithm of~\cite{Motwani&2007}.
We prove $\Omega\big(\|\pi\|^{-1}\big)$ samples are necessary, too, to get a fair estimate of $\gamma$.

We then employ \sumest\ to build \taupiest, a randomized algorithm for approximating $\pi(v)$.
Random-walk-based sampling and time reversibility are the ingredients that allow one to make the connection.
We prove:
\begin{restatable}{theorem}{thmtaupi}
\label{thm:ub_taupi}
Given any $\delta,\epsilon \in (0,1)$ and any state $v$ in a time-reversible Markov chain, \taupiest($\epsilon, \delta, v$) with probability $(1-\delta)$ returns a multiplicative $(1 \pm \epsilon)$ approximation of $\pi(v)$ using $\tilde{O}(\tau\|\pi\|^{-1} \epsilon^{-3} (\ln{\frac{1}{\delta}})^{3/2}) = \tilde{O}(\tau\|\pi\|^{-1})$ elementary operations and calls to \step() and \probe().
\end{restatable}
\noindent Previous algorithms work also for general (i.e.\ non-reversible) chains; but on the $n-o(n)$ states with mass $\pi(v) = O(1/n)$, their complexity becomes at least $\tilde{O}(\tau n)$~\cite{Lee&2013} or $\tilde{O}(\tau^{1.5} \sqrt{n})$~\cite{Lofgren&2015b}.
In fact, $\pi(v)$ can be arbitrarily small (even exponentially small in $n$ and $\tau$) for almost all states in the chain, so for almost all states the complexity of previous algorithms blow up while that of \taupiest\ remains unchanged: since $\|\pi\|^{-1} \le \sqrt{n}$ for any $\pi$, the complexity of \taupiest\ is at most $\tilde{O}(\tau \sqrt{n})$.

Next, we show that \taupiest\ is optimal as a function of $\tau$, $n$ and $\|\pi\|$, up to small factors.
In fact, no algorithm can perform better even if equipped with an operation \neigh{$u$} that returns all incoming and outgoing transition probabilities of $u$.
Formally, we prove:
\begin{restatable}{theorem}{thmlbtaupi}
\label{thm:lb_taupi}
For any function $\nu(n) \in \Omega(1/\sqrt{n}) \cap O(1)$ there is a family of time-reversible chains on $n$ states where (a) $\|\pi\|=\Theta(\nu(n))$, and (b) there is a target state $v$ such that, to estimate its mass $\pi(v)$ within any constant multiplicative factor with constant probability, any algorithm requires $\Omega(\tau \|\pi\|^{-1} / \ln{n})$ \neigh{} calls where $\tau$ is the mixing time of the chain.
\end{restatable}

Although bounding time complexity is our primary goal, in some scenarios one wants to bound the \textit{footprint}, i.e.\ the number of distinct states visited.
Obviously, the footprint of \taupiest\ is bounded by its complexity (Theorem~\ref{thm:ub_taupi}).
We give a second algorithm, \taunest, whose footprint can be smaller than that of \taupiest\ depending on $\tau, n$, and $\|\pi\|$.
More precisely, we prove a footprint bound that is conditional on the concentration of the footprint itself (see Subsection~\ref{sub:fma} for the intuition behind it).
\begin{restatable}{theorem}{thmtaun}
\label{thm:ub_taun}
Let $N_{v,T}$ be the number of distinct states visited by a random walk of $T$ steps starting from $v$.
Assume for a function $\bar{\tau}$ of the chain we have $\prob[ N_{v,T} \notin \Theta(\E[N_{v,T}])] = o\big(\frac{\bar{\tau}}{\E[N_{v,T}]})$.
Then, given any $\delta,\epsilon \in (0,1)$, with probability $(1-\delta)$ one can obtain a multiplicative $(1\pm\epsilon)$-approximation of $\pi(v)$ by visiting $O(f(\epsilon,\delta)(\tau \ln n + \sqrt{\bar{\tau} n}))$ distinct states.
\end{restatable}

\noindent If in Theorem~\ref{thm:ub_taun} we have $\bar{\tau} = \tau$, then \taunest\ is essentially optimal too. Formally:
\begin{restatable}{theorem}{thmlbtaun}
\label{thm:lb_taun}
For any function $\tau(n) \in \Omega(\ln{n}) \cap O(n)$ there is a family of time-reversible chains on $n$ states where (a) the mixing time is $\tau=\Theta(\tau(n))$, and (b) there is a target state $v$ such that, to estimate its mass $\pi(v)$ within any constant multiplicative factor with constant probability, any algorithm requires $\Omega(\sqrt{\tau n/\ln{n}})$ \neigh{} calls.
\end{restatable}

\section{Estimating sum by weighted sampling}
\label{sec:sumapprox}
In this section we analyse the following problem.
We are given a vector of nonnegative reals $\gamma_u$ indexed by the elements $u$ of a set $V$.
The vector is unknown, including its length, but we can draw samples from $V$ according to the distribution $\pi$ where $u$ has probability $\gamma_u / \sum_{u \in V}\!\gamma_u$.
The goal is to approximate the vector sum $\gamma = \sum_{u \in V}\!\gamma_u$.
We describe a simple randomized algorithm, \sumest, which proceeds by repeatedly drawing samples and checking for repeats (i.e.\ a draw that yields an element already drawn before).
The key intuition is the following: at any instant, if $S \subseteq V$ is the subset of elements drawn so far, then the next draw is a repeat with probability $\sum_{u \in S} \gamma_u / \gamma$.
By drawing a sequence of samples we can thus flip a sequence of binary random variables, each one telling if a draw is a repeat, whose expectation is known save for the factor $1/\gamma$.
If the sum of these random variables is sufficiently close to its expectation, one can then get a good approximation of $\gamma$ by simply computing a ratio.
The code of \sumest\ is listed below.

\renewcommand{\thealgorithm}{}
\begin{algorithm}[h!]
\caption{\sumest($\epsilon, \delta$)}
\begin{algorithmic}[1]
\small
\State $S \leftarrow \emptyset$ \Comment{distinct elements drawn so far}
\State $w_S\leftarrow 0$ \Comment{$\sum_{u \in S} \gamma_u$ for the current $S$}
\State $w\leftarrow 0$ \Comment{cumulative sum of $\sum_{u \in S} \gamma_u$ so far}
\State $r\leftarrow 0$ \Comment{number of repeats so far}
\State $\kerr \leftarrow \lceil\frac{2+4.4\epsilon}{\epsilon^2}\ln{\!\frac{3}{\delta}}\rceil$ \Comment{halting threshold on the number of repeats}
\vspace*{0.4em}
\While{$r < \kerr$}
\State $w\leftarrow w + w_S$
\State $(u,\gamma_u) \leftarrow $ sample drawn from distribution $\pi$
\If{$u \in S$} \Comment{detect collision}
  \State $r \leftarrow r+1$
\Else 
  \State $S \leftarrow S \cup \{u\}$
  \State $w_S \leftarrow w_S + \gamma_u$
\EndIf
\EndWhile
\State \textbf{return} $w/r$ \Comment{estimate of $\gamma$}
\end{algorithmic}
\end{algorithm}

We prove:
\begin{theorem}
\label{thm:sumest_approx}
\sumest($\epsilon, \delta$) with probability at least $1 - \frac{2\delta}{3}$ returns an estimate $\hat{\gamma}$ such that $|\hat{\gamma} - \gamma| < \epsilon \gamma$.
\end{theorem}
\begin{proof}
\label{apx:proof_sumest_approx}
We make use of a martingale tail inequality originally from~\cite{Freedman1975} and stated (and proved) in the following form as Theorem 2.2 of~\cite{Alon&2010}, p.~1476:
\begin{theorem}[\cite{Alon&2010}, Theorem 2.2]
\label{thm:alon}
Let $(Z_0,Z_1,\ldots)$ be a martingale with respect to the filter $(\mathcal{F}_i)$. Suppose that $Z_{i+1}-Z_i \le M$ for all $i$, and write $V_t = \sum_{i=1}^t Var(Z_i|\mathcal{F}_{i-1})$. Then for any $z,v>0$ we have
\[
\prob\big[Z_t \ge Z_0 + z, V_t \le v \text{ for some } t\big] \le \exp\!{\Big[\!-\frac{z^2}{2(v+Mz)}\Big]}
\]
\end{theorem}
Let us plug into the formula of Theorem~\ref{thm:alon} the appropriate quantities from \sumest:
\begin{itemize}\itemsep0pt
\item Let $X_i$ be the $(i+1)^{th}$ sample (i.e.\ $(X_i, \gamma_{X_i})$ is the pair $(u,\gamma_u)$ drawn at the $(i+1)^{th}$ invocation of line 8).
\item Let $\mathcal{F}_i$ be the event space generated by $X_0, \dots, X_i$, so that for any random variable $Y$, with $\E[Y|\mathcal{F}_i]$ we mean $\E[Y|X_0,\dots,X_i]$ and with $Var[Y|\mathcal{F}_i]$ we mean $Var[Y|X_0,\dots,X_i]$.
\item Let $\chi_i = \mathbbm{1}[X_i \in \bigcup_{j=0}^{i-1}\{X_j\}]$ be the indicator variable of a repeat on the $(i+1)^{th}$ sample.
\item Let $P_i=\sum_{u\in \cup_{j=0}^{i-1} \{X_j\}}\! \frac{\gamma_{u}}{\gamma}$ be the probability of a repeat on the $(i+1)^{th}$ sample as a function of all the (distinct) samples up to the $i^{th}$,
i.e.\ $P_i=\E[\chi_i|\mathcal{F}_{i-1}] \le 1$.
\item Let $Z_i=\sum_{j=0}^i (\chi_j - P_{j})$;
it is easy to see that $(Z_i)_{i \ge 0}$ 
is a martingale with respect to the filter $(\mathcal{F}_i)_{i\ge 0}$, since $Z_i$ is obtained by adding to $Z_{i-1}$ the indicator variable $\chi_i$ and subtracting $P_i$ i.e.\ its expectation in $\mathcal{F}_{i-1}$.
More formally, $\E[Z_i|\mathcal{F}_{i-1}]=\E[Z_{i-1} + \chi_i- P_i|\mathcal{F}_{i-1}]$, and since $Z_{i-1}$ and $P_i$ are completely determined by $X_0,\dots,X_{i-1}$,
the right-hand term is simply $Z_{i-1} + (\E[\chi_i|\mathcal{F}_{i-1}]-P_i) = Z_{i-1}$. Note also that $Z_0 = 0$.
\item Let $M=1$, noting that $|Z_{i+1}-Z_i| = |\chi_{i+1} - P_{i+1}| \le 1$ for all $i$.
\end{itemize}
\vspace{5pt}
Finally, note that $Var(Z_j|\mathcal{F}_{j-1}) = Var(\chi_j|\mathcal{F}_{j-1})$ (as $Z_j = Z_{j-1} + \chi_j - P_j$ and, again, $Z_{j-1}$ and $P_j$ are completely determined by $X_0,\dots,X_{j-1}$).
Since $Var(\chi_j|\mathcal{F}_{j-1}) = P_j(1-P_j) \le P_j$, we have $V_i = \sum_{j=1}^i Var(Z_j|\mathcal{F}_{j-1}) \le \sum_{j=1}^i P_j$.
Theorem~\ref{thm:alon} then yields the following:
\begin{corollary}
\label{cor:martbound}
For all $z,v > 0$ we have
\begin{align}
\prob\big[Z_i \ge z, \sum_{j=1}^i P_j \le v \text{ for some } i\big] &\le 
\exp{\!\Big[\!-\frac{z^2}{2(v+z)}\Big]}
\end{align}
\end{corollary}
Recall now \sumest.
Note that  $\sum_{j=1}^i P_j$ and $Z_i$ are respectively the value of $\frac{w}{\gamma}$ and of $r-\frac{w}{\gamma}$ just after the \emph{while} loop has been executed for the $(i+1)$-th time.
Note also that, when \sumest\ returns, $r=\kerr$.
Therefore the event that, when \sumest\ returns, $\frac{w}{r} \le \gamma(1-\epsilon)$ i.e.\ $\frac{w}{\gamma} \le r (1-\epsilon) \le (1-\epsilon)\kerr$ corresponds to the event that $Z_i \ge \epsilon r = \epsilon \kerr$ and $\sum_{j=1}^i P_j \le (1-\epsilon)\kerr$.
Invoking Lemma~\ref{cor:martbound} with $z=\epsilon\kerr$ and $v=(1-\epsilon)\kerr$:
\begin{align}
\prob\big[\frac{w}{r} \le \gamma(1-\epsilon)\big]  &\le \exp{\!\Big[\!-\frac{\epsilon^2 k_{\epsilon,\delta}^2}{2(\epsilon k_{\epsilon,\delta} + (1-\epsilon)k_{\epsilon,\delta})}\Big]}
= \exp{\!\Big[\!-\frac{\epsilon^2 k_{\epsilon,\delta}}{2}\Big]}
\end{align}
which is smaller than $\delta/3$ since clearly $\kerr > \frac{2}{\epsilon^2}\ln{\frac{3}{\delta}}$.
Consider instead the event that, when \sumest\ returns, $\frac{w}{r} \ge \gamma(1+\epsilon)$ i.e.\ $\frac{w}{\gamma} \ge r (1+\epsilon) = \kerr(1+\epsilon)$.
This is the event that $Z_i \le -\epsilon\kerr$, or equivalently $-Z_i \ge \epsilon\kerr$.
Note that Lemma~\ref{cor:martbound} still holds if we replace $Z_i$ with $-Z_i$, as $(-Z_i)_{i \ge 0}$ too is obviously a martingale with respect to the filter $(\mathcal{F}_i)_{i\ge 0}$, with $-Z_0=0$. Let then $i_0 \le i$ be the smallest time such that $-Z_{i_0} \ge \epsilon k_{\epsilon,\delta}$.
Since $|Z_j - Z_{j-1}| \le 1$, it must be $-Z_{i_0} < \epsilon k_{\epsilon,\delta}+1$.
Also, since $\sum_{j=0}^i \chi_j$ is nondecreasing with $i$, then $\sum_{j=0}^{i_0} \chi_j \le k_{\epsilon,\delta}$.
It follows that $\sum_{j=1}^{i_0} P_j = -Z_{i_0} + \sum_{j=0}^{i_0} \chi_j \le \epsilon k_{\epsilon,\delta} + 1 + k_{\epsilon,\delta} = (1+\epsilon)k_{\epsilon,\delta}+1$.
Invoking again Lemma~\ref{cor:martbound} with $z=\epsilon k_{\epsilon,\delta}$ and $v=(1+\epsilon)k_{\epsilon,\delta}+1$, we obtain:
\begin{align}
\prob\big[\frac{w}{r} \ge \gamma(1+\epsilon)\big] 
&\le \exp{\!\Big[\!-\frac{\epsilon^2 k_{\epsilon,\delta}^2}{2((1+2\epsilon)k_{\epsilon,\delta}+1)}\Big]}
\end{align}
Note that $\frac{1}{\kerr} < \frac{\epsilon^2}{2+4.4\epsilon} < 0.2\epsilon$ since $\epsilon \le 1$; so $2((1+2\epsilon)+\frac{1}{\kerr}) < 2+4.4\epsilon$, and since $\kerr \ge \frac{2+4.4\epsilon}{\epsilon^2}\ln{\!\frac{3}{\delta}}$ the right-hand term is at most $\frac{\delta}{3}$.
Finally, by a simple union bound the probability that $|\hat{\gamma} - \gamma| \ge \epsilon \gamma$  is at most $2\frac{\delta}{3}$, and the proof of Theorem~\ref{thm:sumest_approx} is complete.
\end{proof}
\begin{theorem}
\label{thm:sumest_cost}
\sumest($\epsilon,\delta$) draws at most
$\lceil 45 \|\pi\|^{-1} \epsilon^{-3}(\ln{\frac{3}{\delta}})^{3/2} \rceil$ samples with probability at least $1-\frac{\delta}{3}$.
\end{theorem}
\begin{proof}
\label{apx:proof_sumest_cost}
\newcommand{\myb}{c}
\newcommand{\mybb}{\frac{5}{18}}
\newcommand{\mys}{s}
\newcommand{\myss}{\bar{s}}
\newcommand{\mySS}{\bar{S}}
We show that the probability that $\mys = \lceil 45 \|\pi\|^{-1} \epsilon^{-3}(\ln{\frac{3}{\delta}})^{3/2} \rceil$ draws yield less than $\kerr$ repeats is less than $\frac{\delta}{3}$.
Let $\bar{p} = \mybb\|\pi\| \epsilon (\ln{\frac{3}{\delta}})^{-1/2}$.
We consider two cases.
\\\textbf{Case 1}: $\exists u \in V$ with $\pi(u) > \bar{p}$.
Let then $C_u^s$ be the random variable counting the number of times $u$ appears in $\mys$ draws.
Since if $C_u^s > \kerr$ then $u$ causes at least $\kerr$ repeats, the probability that \sumest\ needs more than $s$ draws is upper bounded by $\prob[C_u^s \le \kerr]$.
Now $\E[C_u^s] = \mys \pi(u) > \mys \bar{p} >  45\mybb \frac{1}{\epsilon^2}\ln{\frac{3}{\delta}} = \frac{12.5}{\epsilon^2}\ln{\frac{3}{\delta}} \ge 1.7 (\frac{6.4}{\epsilon^2}\ln{\frac{3}{\delta}} +1) \ge 1.7 \lceil \frac{2 + 4.4\epsilon}{\epsilon^2} \ln{\frac{3}{\delta}} \rceil = 1.7 \kerr$, therefore
$C_u^s \le \kerr$ implies $C_u^s < \frac{1}{1.7}\E[C_u^s] < (1 - 0.41)\E[C_u^s]$.
Since $C_u^s$ is a sum of independent binary random variables, the bounds of Appendix~\ref{apx:chernoff_bounds} give $\prob[C_u^s \le \kerr] < \exp{\!\big(\!-\!\frac{1}{2} 0.41^2 \E[C_u^s] \big)} < \exp{\!\big(\!-0.5 \cdot 0.41^2 \cdot \frac{12.5}{\epsilon^2} \ln\frac{3}{\delta} \big)} < \exp{\!\big(\!-1.05 \ln\frac{3}{\delta} \big)} < \frac{\delta}{3}$.
\\\textbf{Case 2}: $\pi(u) \le \bar{p}$ for all $u \in V$.
Let then $\myss = \lceil \bar{p}^{-1} \rceil$, let $\mySS$ be the set of distinct elements in the first $\myss$ draws, and let $w(\myss) = \sum_{u \in \mySS} \pi(u)$.
First we show that $\E[w(\myss)] \ge \frac{4}{9}\myss\|\pi\|^2$.
Write $\E[w(\myss)] = \sum_{u \in V} \pi(u)(1-(1-\pi(u))^{\myss})$.
Since for all $x \in [0,1]$ and $k \ge 1$ it holds $(1-x)^k \le (1+kx)^{-1}$, by setting $x=\pi(u)$ and $k=\myss$ we obtain $1-(1-\pi(u))^{\myss} \ge 1-(1+\myss\pi(u))^{-1} = \myss\pi(u)(1+\myss\pi(u))^{-1}$.
Moreover note that $\bar{p}^{-1} \ge \frac{18}{5} = 3.6$ and thus $\lceil\bar{p}^{-1}\rceil \le \frac{5}{4}\bar{p}^{-1}$. Therefore $\myss \pi(u) \le \lceil\bar{p}^{-1}\rceil\bar{p} \le \frac{5}{4}$ for all $u$, and thus $\myss\pi(u)(1+\myss\pi(u))^{-1} \ge \myss\pi(u)\frac{1}{1+\frac{5}{4}} = \frac{4}{9}\myss\pi(u)$.
Therefore $\E[w(\myss)] \ge \frac{4}{9} \myss \sum_{u \in V} \pi(u)^2 = \frac{4}{9}\myss\|\pi\|^2$.
Now we consider two cases.
First, suppose the event $w(\myss) \ge 0.4\,\E[w(\myss)]$ takes place.
For $i=\myss+1,\ldots,\mys$ let $\chi_i$ be the indicator random variable of the event that the $i$-th draw is an element of $\mySS$, and let $C_s = \sum_{i=\myss+1}^{\mys}\chi_i$.
Clearly \sumest\ witnesses at least $C_s$ repeats in the last $\mys - \myss$ draws, and thus overall.
We shall then bound $\prob[C_s < \kerr]$.
First, since by hypothesis the total mass of $\mySS$ is $w(\myss) \ge 0.4\, \E[w(\myss)]$, we also have $\E[\chi_i] \ge 0.4\, \E[w(\myss)] \ge \frac{1.6}{9} \myss \|\pi\|^2$.
Therefore $\E[C_s] =  \sum_{i=\myss+1}^{\mys}\E[\chi_i] \ge \frac{1.6}{9}(\mys-\myss) \myss \|\pi\|^2$.
Now note that $\mys-\myss > 10 \myss$, therefore $\E[C_s] \ge \frac{16}{9} \myss^2 \|\pi\|^2$.
Finally, since $\myss = \lceil \bar{p}^{-1} \rceil \ge \frac{18}{5} \|\pi\|^{-1} \epsilon^{-1} (\ln{\frac{3}{\delta}})^{1/2}$, it holds $\E[C_s] \ge (\frac{18}{5})^2\frac{16}{9}\frac{1}{\epsilon^{2}}\ln{\frac{3}{\delta}} > 23 \frac{1}{\epsilon^{2}}\ln{\frac{3}{\delta}} > 3.14\kerr$.
It follows that the event $C_s < \kerr$ implies $C_s < \frac{1}{3.14}\E[C_s] < (1-0.68)\E[C_s]$.
By the concentration bounds of Appendix~\ref{apx:chernoff_bounds}, the probability of the latter is $\prob[C_s < \kerr] \le \exp{\!\big(\!-\!\frac{1}{2}\,0.68^2\,\E[C_s] \big)} < \exp{\!\big(\!-\frac{1}{2} \, 0.68^2 \, 23 \frac{1}{\epsilon^{2}}\ln{\frac{3}{\delta}}\big)} < \exp{\!\big(\!-\!5\ln{\frac{3}{\delta}} \big)}  < \frac{\delta}{243}$.
The second case corresponds to the event $w(\myss) < 0.4 \, \E[w(\myss)] = (1-0.6)\E[w(\myss)]$, of which we shall bound the probability.
Let $\chi_u^{\myss}$ be the indicator variable of the event $u \in \mySS$, so $w(\myss) = \sum_{u \in V} \chi_u^{\myss} \, \pi(u)$.
Since $\pi(u) \le \bar{p}$ for all $u$, we can write $w(\myss) = \bar{p} \sum_{u \in V} \chi_u^{\myss} \,\,\bar{p}^{-1}\pi(u)$ so that the coefficients $\bar{p}^{-1}\pi(u)$ are in $[0,1]$.
Clearly, the $\chi_u^{\myss}$ are non-positively correlated.
We can thus apply the bounds of Appendix~\ref{apx:chernoff_bounds} and get $\prob[w(\myss) < 0.4\, \E[w(\myss)]] \le \exp{\!\big(\!- 0.5 \cdot 0.6^2 \, \bar{p}^{-1}\E[w(\myss)] \big)}$.
By replacing the definitions and bounds for $E[w(\myss)]$, $\myss$ and $\bar{p}^{-1}$ from above, we get $\prob[w(\myss) < 0.4\, \E[w(\myss)]] < \exp{\!\big(\!- 2.88 \ln{(\frac{3}{\delta})} \big)} < \frac{\delta}{23}$.
Again by a union bound, the probability that \sumest\ draws more than $\mys$ samples is less than $\frac{\delta}{243} + \frac{\delta}{23} < \frac{\delta}{3}$.
\end{proof}
We remark that the previous existing algorithm for the sum estimation problem~\cite{Motwani&2007} needs knowledge of $n = |V|$ and uses $O(\sqrt{n} \epsilon^{-7/2} \log(n)(\log{\frac{1}{\delta}} + \log{\frac{1}{\epsilon}} + \log\log{n}))$ samples.
\sumest\ is simpler, oblivious to $n$, and gives more general bounds.
It is also asymptotically faster unless $\pi$ is (essentially) the uniform distribution.

Finally, we show that \sumest\ is essentially optimal, by proving $\Omega(\|\pi\|^{-1})$ samples are in general necessary to estimate $\gamma$ even if $n$ is known in advance.
This extends to arbitrary distributions the $\Omega(\sqrt{n})$ lower bound given by~\cite{Motwani&2007} for the uniform distribution.
\begin{restatable}{theorem}{lbsumest}
\label{thm:lbsumest}
For any function $\nu(n) \in \Omega(n^{-\frac{1}{2}}) \cap O(1)$ there exist vectors $\mathbf{x} = \gamma \pi = (\gamma_1,\ldots,\gamma_n)$ with $\|\pi\| = \Theta(\nu(n))$ such that $\Omega(\|\pi\|^{-1})$ samples are necessary to estimate $\gamma$ within constant multiplicative factors with constant probability, even if $n$ is known.
\end{restatable}
\begin{proof}
\label{apx:sumlb}
Let $k \in \Theta(\nu(n)^{-2})$ with $1 \le k \le \frac{n}{2}$.
Consider the two vectors $\mathbf{x} = (\gamma_1,\ldots,\gamma_n)$ and $\mathbf{x}'= (\gamma_1',\ldots,\gamma_n')$ defined as follows:
\begin{align*}
&\gamma_j = 1 : j\le k, \qquad \gamma_j=\sqrt{k}/n : j>k\\
&\gamma_j' = 1 : j\le 2k \qquad \gamma_j'=\sqrt{k}/n : j>2k
\end{align*}
Now let $\gamma = \sum_{i=1}^n \gamma_i$ and $\gamma' = \sum_{i=1}^n \gamma_i'$.
One can check that $\gamma \le 2k$ and $|\gamma - \gamma'| \ge \frac{k}{2}$.
Hence, to obtain an estimate $\hat{\gamma}$ of $\gamma$ with $\hat{\gamma} \le \frac{5}{4}\gamma$, one must distinguish $\mathbf{x}$ from $\mathbf{x}'$.
Note that the norm of $\pi = \mathbf{x} / \gamma$ is in $\Theta(1/\sqrt{k}) = \Theta(\nu(n))$, as requested.
Now, for each one of $\mathbf{x}$ and $\mathbf{x}'$ in turn, pick a permutation of $\{1,\ldots,n\}$ uniformly at random and apply it to the entries of the vector.
Suppose then we sample $o(\|\pi\|^{-1}) = o(\sqrt{k})$ entries from $\mathbf{x}$.
We shall see that, with probability $1-o(1)$, we cannot distinguish $\mathbf{x}$ from $\mathbf{x}'$.
First, note that the total mass of the entries with value $\sqrt{k}/n$ is at most $1/\sqrt{k}$.
Hence the probability of drawing \emph{any} of those entries with $o(\sqrt{k})$ samples is $o(1)$, and we can assume all draws yield entries having value $1$.
Since there are $O(k)$ such entries in total, the probability of witnessing any repeat is also $o(1)$, and we can assume no repeat is witnessed.
Furthermore, because of the random permutation, the indices of samples are distributed uniformly over $\{1,\ldots,n\}$ (recall that we actually sample from the index set $\{1,\ldots,n\}$, so we could use the distribution of the indices to distinguish $\mathbf{x}$ from $\mathbf{x}'$).
The same argument applies to $\mathbf{x}'$, so drawing $o(\sqrt{k})$ samples from $\mathbf{x}'$ yields exactly the same distribution and the two vectors are indistinguishable.
To adapt the construction to larger approximation factors, set $\gamma_j' = 1 : j\le \eta k$ for $\eta$ large enough.
\end{proof}

\section{Approximating the stationary distribution}
\label{sec:massapprox}
In this section we address the problem of approximating $\pi(v)$.
Such a problem can in fact be reduced to the sum estimation problem of Section~\ref{sec:sumapprox} by drawing states via random walks.
The crux is determining how long the walks must be in order for the samples to come from a distribution close enough to $\pi$, so that the approximation guarantees of \sumest\ transfer directly to our estimate of $\pi(v)$.

Consider a random walk of length $t+1$ that starts at $v$.
Obviously we can simulate such a walk by setting $u_0 = v$ and then invoking \step($u_i$) to obtain the state $u_{i+1}$, for $i=0,\ldots,t-1$.
Crucially, using the time-reversibility of the chain, for any visited state $u$ we can obtain the ratio $\gamma_u$ between $\pi(u)$ and $\pi(v)$ using $O(1)$ operations.
Formally, let $\gamma_v = \pi(v)/\pi(v) = 1$, and in general let $\gamma_u = \pi(u)/\pi(v)$.
Note that:
\begin{align}
\label{eqn:computegamma}
\gamma_{u_{i+1}} = \frac{\pi(u_{i+1})}{\pi(v)} = \frac{\pi(u_{i+1})}{\pi(u_{i})} \cdot \frac{\pi(u_{i})}{\pi(v)}
= \frac{\pi(u_{i+1})}{\pi(u_{i})} \cdot \gamma_{u_i}
\end{align}
The time-reversibility of the chain (see Equation~\ref{eqn:detbalance}) implies $\frac{\pi(u_{i+1})}{\pi(u_{i})} = \frac{p_{u_{i},u_{i+1}}}{p_{u_{i+1},u_{i}}} = \frac{\probe(u_{i},u_{i+1})}{\probe(u_{i+1},u_{i})}$, hence we can compute $\gamma_{u_{i+1}}$ with $O(1)$ operations if we know $\gamma_{u_i}$.
But then we can keep track of $\gamma_u$ for any $u$ visited so far, starting with $\gamma_{u_0} = 1$ and computing $\gamma_{u_{i+1}}$ by Equation~\ref{eqn:computegamma} the first time $u_{i+1}$ is visited.

Suppose now to pick $t$ large enough so that the chain reaches its stationary distribution, i.e.\ $u_t \sim \pi$ irrespective of $v$.
One is then drawing state $u$, as well as its associate weight $\gamma_{u}$, with probability $\pi(u)$.
Now if we let $\gamma = \sum_{u \in V} \gamma_u$, then $\pi(u) = \gamma_u \gamma^{-1}$ and in particular $\pi(v) = \gamma^{-1}$.
Therefore approximating $\pi(v)$ amounts to approximating $\gamma$; more formally, for any $\epsilon \in (0,1)$, if $\hat{\gamma}$ is a $(1\pm\frac{\epsilon}{2})$-approximation of $\gamma$ then $\hat{\gamma}^{-1}$ is a $(1\pm \epsilon)$-approximation of $\pi(v)$.
We can therefore reduce to the sum approximation problem of Section~\ref{sec:massapprox}: compute with probability $(1-\delta)$ a $(1\pm\epsilon)$-approximation of $\gamma$, assuming we can draw pairs $(u,\gamma_u)$ according to $\pi$.
The only problem is that by simulating the chain we can only come close to (but not exactly on) the stationary distribution $\pi$.
We must then tie the approximation guarantees of \sumest\ to the length $t$ of the random walks, or better to the distance $\tvd{\pi'}{\pi}$ between $\pi$ and the distribution $\pi'$ from which $u_t$ is drawn.
Formally, we show:
\begin{lemma}
\label{lem:adapt}
There exists some constant $c > 0$ such that the following holds.
Choose any ${\delta,\epsilon \in (0,1)}$, and suppose we draw the pairs $(u,\gamma_u)$ from a distribution $\pi'$ such that $\tvd{\pi}{\pi'} \le \big(\frac{\epsilon \|\pi\|}{\ln(3/\delta)}\big)^c$.
Then \sumest($\frac{\epsilon}{2},\delta$) with probability at least $1-\delta$ returns a multiplicative $(1\pm \epsilon)$-approximation of $\gamma$ by taking at most $\lceil 720 \|\pi\|^{-1} \epsilon^{-3}(\ln{\frac{3}{\delta}})^{3/2} \rceil$ samples.
\end{lemma}
\begin{proof}
Let us start with the bound on the number of samples.
Recall the proof of Theorem~\ref{thm:sumest_cost}, and note that the whole argument depends on $\pi$ but not on the values $\gamma_u$.
Indeed, $\pi$ alone determines the probability of repeats and thus controls the distribution of the number of samples drawn by \sumest.
Hence, by Theorem~\ref{thm:sumest_cost} \sumest($\frac{\epsilon}{2},\delta$) takes more than $\lceil 45 \|\pi'\|^{-1} 8 \epsilon^{-3}(\ln{\frac{3}{\delta}})^{3/2} \rceil = \lceil 360 \|\pi'\|^{-1} \epsilon^{-3}(\ln{\frac{3}{\delta}})^{3/2} \rceil$ samples with probability less than $\frac{\delta}{3}$.
Now $\|\pi - \pi'\| \le 2\tvd{\pi}{\pi'} \le 2\big(\frac{\epsilon \|\pi\|}{\ln(3/\delta)}\big)^c \le \|\pi\| 2(\ln{3})^{-c}$, which for $c \ge 15$ is bounded by $\frac{1}{2} \|\pi\|$.
Then, since $\|\pi'\| \ge \|\pi\| - \|\pi - \pi'\|$, we have $\|\pi'\|^{-1} \le 2 \|\pi\|^{-1}$ and the bound above is in turn bounded by $\lceil 720 \|\pi\|^{-1} \epsilon^{-3}(\ln{\frac{3}{\delta}})^{3/2} \rceil$.

Let us now see the approximation guarantees.
Recall the proof of Theorem~\ref{thm:sumest_approx}.
We want to show again that $\Pr[|\frac{w(s)}{r} - \gamma| \ge \frac{\epsilon}{2}\gamma] \le \frac{2\delta}{3}$.
However, now the samples are drawn according to $\pi'$ instead of $\pi$.
Let then $P_j'=\sum_{u\in \cup_{h=0}^{j-1} \{X_h\}}\! \pi'(u)$ and $Z_i'=\sum_{j=0}^i (\chi_j - P_{j}')$; in a nutshell, $P_j'$ and $Z_i'$ are the analogous of $P_j$ and $Z_i$ under $\pi'$.
It is immediate to check that Lemma~\ref{cor:martbound} holds with $Z_i'$ and $P_j'$ in place of $Z_i$ and $P_j$.
Let now $w'(i) = \gamma \sum_{j=1}^i P_{j}'$.
Note that $\sum_{j=1}^i P_{j}'$ and $Z_i'$ are respectively the value of $\frac{w'(i)}{\gamma}$ and of $r - \frac{w'(i)}{\gamma}$ just after line 9 has been executed for the $(i+1)$-th time.
Therefore, the argument following Lemma~\ref{cor:martbound} holds if we put $w'(i)$ in place of the value taken by $w$ after the $(i+1)$-th execution of line 9.
Hence the same bounds hold, and \sumest($\frac{\epsilon}{2},\delta$) ensures $\Pr[|\frac{w'(s)}{r} - \gamma| \ge \frac{\epsilon}{2}\gamma] \le \frac{\delta}{3}$ where $s$ is the total number of draws.
Now note that \sumest\ does not return $\frac{w'(s)}{r}$, but $\frac{w(s)}{r}$ where $w(i) = \gamma \sum_{j=1}^i P_{j}$ is the value of $w$ in \sumest\ after line 9 has been executed for the $(i+1)$-th time.
We shall now make $|\frac{w(s)}{r} - \frac{w'(s)}{r}| \le \frac{\epsilon}{2}\gamma$; by the triangle inequality we will then be done.
First of all, by the definition of $w(s)$ and $w'(s)$ we have
\begin{equation}
\label{eqn:diffw}
\Big|\frac{w(s)}{r} - \frac{w'(s)}{r}\Big| = \gamma r^{-1} \Big|\sum_{j=1}^s P_{j} - \sum_{j=1}^s P_j'\Big| \le \gamma r^{-1} \sum_{j=1}^s \big|P_{j} - P_j'\big|
\end{equation}
Now note that $|P_j - P_j'| \le \tvd{\pi}{\pi'}$, since $P_j$ and $P_j'$ are the probability of the same event under respectively $\pi$ and $\pi'$.
Therefore the right-hand side of Equation~\ref{eqn:diffw} is bounded by $\gamma r^{-1} s \, \tvd{\pi}{\pi'}$.
Now, when \sumest($\frac{\epsilon}{2},\delta$) terminates $r = k_{\frac{\epsilon}{2},\delta} \ge 4\frac{2 + 2.2\epsilon}{\epsilon^2}\ln\frac{3}{\delta}$, and by hypothesis $\tvd{\pi}{\pi'} \le \big(\frac{\epsilon \|\pi\|}{\ln(3/\delta)}\big)^c$.
Therefore:
\begin{equation}
\Big|\frac{w(s)}{r} - \frac{w'(s)}{r}\Big|
\le \gamma s \, \frac{\epsilon^2}{4(2 + 2.2\epsilon)\ln(\frac{3}{\delta})} \Big(\frac{\epsilon \|\pi\|}{\ln(\frac{3}{\delta})}\Big)^c
\le \gamma \, s \, \|\pi\| \, \frac{\epsilon^{3+c}}{8\ln(\frac{3}{\delta})^{1+c}} 
\end{equation}
Finally, recall from above that with probability $1-\frac{3}{\delta}$ we have $s \le \lceil 720 \|\pi\|^{-1} \epsilon^{-3}(\ln{\frac{3}{\delta}})^{3/2} \rceil$.
In this case the equation above yields $|\frac{w(s)}{r} - \frac{w'(s)}{r}| \le \gamma \cdot 721 \epsilon^c \ln(\frac{3}{\delta})^{0.5-c} $, which is smaller than $\frac{\epsilon}{2}\gamma$ for $c \ge \frac{1}{2}+ \frac{\ln 1442}{\ln \ln 3} \approx 78$.

A simple union bound completes the proof.
\end{proof}
We are now ready to prove Theorem~\ref{thm:ub_taupi}.
Pick $t = \tau \, c \ln{\!(\|\pi\|^{-1}\epsilon^{-1} \ln{\frac{3}{\delta}})} / \ln 2$, where $c$ is the constant of Lemma~\ref{lem:adapt} and $\tau$ is the mixing time of the chain.
Simulate the walk for $t$ steps starting from $v$, and let $\pi'$ be the distribution of the final state.
By the properties of the mixing time (see Section~\ref{sub:notation}):
\begin{align}
\tvd{\pi}{\pi'} \le 2^{-c \ln{\!(\|\pi\|^{-1}\epsilon^{-1} \ln{\frac{3}{\delta}})} / \ln 2} \le \Big(\frac{\epsilon \|\pi\|}{\ln(3/\delta)}\Big)^c
\end{align}
and therefore by Lemma~\ref{lem:adapt} we obtain a $(1\pm\epsilon)$ approximation of $\gamma$.
By choosing $\epsilon$ small enough we can obtain a $(1\pm\epsilon')$ approximation of $\pi(v)$ for any desired $\epsilon'$.
The total number of operations performed is clearly bounded by $t = \tau \, c \ln{\!(\|\pi\|^{-1}\epsilon^{-1} \ln{\frac{3}{\delta}})} / \ln 2$ times the number of samples taken by \sumest, and by substituting this value in the bound of Theorem~\ref{thm:sumapprox} we obtain Theorem~\ref{thm:ub_taupi}.
The pseudocode of the resulting algorithm, \taupiest, is given for reference in Appendix~\ref{apx:taupiest}.

\subsection{Reducing the footprint}
\label{sub:fma}
In this section we describe \taunest, the algorithm behind the bounds of Theorem~\ref{thm:lb_taun}.
\taunest\ is derived from \taupiest\ as follows.
First, instead of performing a new walk of length $t$ from $v$ for each sample, the algorithm performs one long random walk of length $T$ and takes one sample every $t$ steps.
The correctness guarantees do not change, since although the samples do not come all from the same distribution, they are still drawn from a distribution sufficiently close to $\pi$.
Second, after checking if the current draw yields a repeat, the algorithm includes in the set $S$ not only the draw but also all other states visited so far.
Again, this does not affect the guarantees, since we do not need the set $S$ to be built on independent samples.
However, this makes the mass of $S$ grow potentially faster, so we can hope to get more repeats and decrease the total number of samples.
The pseudocode of \taunest\ is in Appendix~\ref{apx:fma}.

\textbf{The concentration hypothesis.}
Before continuing to the proof of Theorem~\ref{thm:lb_taun}, let us provide some intuition behind the concentration hypothesis.
Suppose the walk runs for $T = k \bar{\tau}$ steps for some $\bar{\tau} = \tau \poly(\log(\|\pi\|^{-1}))$.
Such a process can be seen as a coupon collector over $k$ rounds, where a subset of at most $\bar{\tau}$ states is collected (i.e.\ visited) at each round.
Now, if we pick $\bar{\tau}' \le \bar{\tau}$ with $\bar{\tau}' = \tau \poly(\log(\|\pi\|^{-1}))$, then in each round the $\bar{\tau} - \bar{\tau}'$ central steps are essentially independent of other rounds (more formally, the correlation is $O(\poly(n)^{-1})$).
Each round is then in large part independent of the others; the issue is that the states visited \textit{within} a single round are correlated.
Such a correlation is responsible for the factor $\bar{\tau}$ in the concentration hypothesis and amounts for the (intuitive) fact that conditioning on the outcome of one step of the walk does not affect the distribution of those steps that are more than $\bar{\tau}$ steps away.
We note that the concentration bounds of~\cite{Chung&2012} give $\prob[\sum_{i=1}^T f_i \notin (1 \pm \bar{\epsilon})\E[\sum_{i=1}^T f_i]] < 2\exp{\!-\Omega\big(\bar{\epsilon}^2 \E[\sum_{i=1}^T f_i] / \tau\big)}$ where $f_i \in [0,1]$ is a function of state $X_i$; however we could not use them to prove the concentration hypothesis of Theorem~\ref{thm:ub_taun}.

Let us now delve into the proof.

\begin{proof}
Observe the random walk performed by \taunest.
Clearly if $\bar{\tau} = \Omega(n)$ then the walk visits $O(\tau\ln n  + \sqrt{\bar{\tau}n})$ distinct states, and the theorem holds unconditionally.
Let us then assume $\bar{\tau} = o(n)$.
We disregard the first $T_0 = \Theta(\tau\ln n)$ steps of the walk, which of course yield at most $T_0$ distinct states, and focus on the last $T$ steps, which we denote by $X_1,\ldots,X_T$ (one may thus plug $T+T_0$ in place of $T$ in the concentration hypothesis).
Let $\pi_i$ denote the distribution of state $X_i$, $i=1,\ldots,T$.
Since $T_0 = \Theta(\tau\ln n)$, then we can make $\tvd{\pi_i}{\pi} \le \frac{1}{\poly(n)}$.
One can adapt the proof of Lemma~\ref{lem:adapt} to \taunest, using the hypothesis $\tvd{\pi}{\pi_i} \le \big(\frac{\epsilon \|\pi\|}{\ln(3/\delta)}\big)^c$ for all $i \ge 1$.
This changes the bounds of the lemma only by constant multiplicative factors.
We can thus focus on proving the bound on the number of states visited by the walk.
In the analysis we assume $X_i \sim \pi$, but again the same asymptotic bounds hold if $\tvd{\pi_i}{\pi} \le \frac{1}{\poly(n)}$.
Let $S_{v,t} = \cup_{i=1}^t \{X_i\}$, let $N_{v,t} = |S_{v,t}|$, and let $M_{v,t} = \sum_{u \in S_{v,t}} \pi(u)$.
For brevity we simply write $S_t, N_t, M_t$.

The crux is to show that $M_t$, the aggregate mass of $S_t$, grows basically as $N_t^2/t$.
Formally we prove that, for any $\epsilon,\delta, q > 0$, if $\prob[N_t \ge q] \ge 1-\delta$ then $\prob[M_t \ge q^2\frac{\epsilon}{4tn}] \ge 1-\epsilon - \delta$.
First, for any $\lambda > 0$ let $V_{\lambda} = \{u \in V : \pi(u) < \frac{\lambda}{n} \}$.
Clearly $\prob[X_i \in V_{\lambda}] = \sum_{u \in V_{\lambda}} \pi(u) < \lambda$.
Therefore the number of steps $J_t(\lambda)$ the chain was on a state of $V_{\lambda}$ satisfies $\E[J_t(\lambda)]< t \lambda$.
Now, by Markov's inequality $\prob[J_t(\lambda) > \frac{q}{2}] < \frac{2 t \lambda}{q}$, and setting $\lambda=\epsilon \frac{q}{2} t$ we obtain $\prob[J_t(\lambda) > \frac{q}{2}] < \epsilon$.
Since by hypothesis $\prob[N_t < q] < \delta$, by a union bound we get $\Pr[N_t \ge q, \, J_t(\lambda) < \frac{q}{2}] \ge 1 - \delta - \epsilon$.
But if $N_t \ge q$ and $J_t(\lambda) < \frac{q}{2}$ then $S_t$ contains at least $\frac{q}{2}$ distinct states with individual mass at least $\frac{\epsilon q}{2tn}$, and thus $M_t \ge \frac{q}{2} \frac{\epsilon q}{2tn} = q^2\frac{\epsilon}{4tn}$.

Now choose $t$ such that $\E[N_{t}] = \Omega(\sqrt{n \bar{\tau}})$; note that $\E[N_{t}] = \Omega(\bar{\tau})$ since $\bar{\tau} = o(n)$.
By plugging $\E[N_{t}]$ into the concentration bound for $N_t$ we can then make $\prob[N_{t} < (1-\bar{\epsilon})\E[N_{t}]]$ arbitrarily small for any $\bar{\epsilon} > 0$.
Let then $q = (1-\bar{\epsilon})\E[N_{t}]$.
By the bounds of the previous paragraph, for any $\delta > 0$ with probability $1-\delta-\bar{\epsilon}$ we have $M_t \ge q^2\frac{\bar{\epsilon}}{4tn} = \Omega(n \bar{\tau}) \frac{\bar{\epsilon}}{4tn} = \Omega(\frac{\bar{\tau}}{t})$.
Conditioned on the event that $M_{t} = \Omega(\frac{\bar{\tau}}{t})$, any sample drawn after $t$ steps is a repeat with probability $\Omega(\frac{\bar{\tau}}{t})$.
If we then draw $\Theta(\frac{t}{\bar{\tau}})$ samples, which require $\Theta(t)$ steps, we witness an expected $\Omega(1)$ samples, which can be made larger than $\kerr$ by appropriately increasing $t$.
Again by the concentration bounds on $N_t$, the total number of states visited can be made $O(2\E[Q_t]) = O(\sqrt{n \bar{\tau}}) = \tilde{O}(\sqrt{n \tau})$ with probability arbitrarily close to $1$ by appropriately increasing $t$.
\end{proof}

\section{Lower bounds}
\label{sec:lb}
In this section we prove the bounds of Theorem~\ref{thm:lb_taupi} and Theorem~\ref{thm:lb_taun} (see Section~\ref{sub:results}).
Both proofs follow the same line.
As anticipated, the bounds are proven under a strengthened model providing a primitive \neigh{$u$} that at cost $O(1)$ returns all the incoming and outgoing transition probabilities of $u$.
We assume \neigh{u} is invoked automatically when $u$ is first visited, and we leave for free all subsequent \step() and \probe() calls on $u$ and all elementary operations.
It is clear that the cost incurred under this model is no larger than that incurred in the \step() and \probe() model.
We see the chains as random walks on weighted undirected graphs.
Recall that any undirected weighted graph $G$ can be univocally associated to a time-reversible Markov chain: for any $u,u' \in G$, $p_{uu'} > 0$ if and only if $(u,u')$ is an edge of $G$ with weight $w_{uu'} = z \pi(u)p_{uu'}$, for some constant $z >0$ equal for all edges.

Consider a random $d$-regular expander graph $G_0$ on $n_0$ nodes.
By standard results, the simple random walk on $G_0$ has mixing time $\tau_0 = \Theta(\log_d{n_0})$.
Also, by standard birthday arguments, starting from any given node $v$ any algorithm must visit $\Omega(\sqrt{n_0})$ nodes and thus perform $\Omega(\sqrt{n_0})$ queries to estimate $n_0$ within constant factors with constant probability.
We now build our chain out of $G_0$.
In particular, we create a (random) graph $G=G(\Delta,n_0)$ on $n = \Theta(n_0 \Delta)$ nodes as follows.
Take each arc $\{u,v\}$ of $G_0$ and replace it with a star on $\Delta+1$ nodes.
More precisely: delete $\{u,v\}$, add a new node $s_{uv}$ and two arcs $\{u,s_{uv}\}$ and $\{v,s_{uv}\}$, and add $\Delta-2$ nodes each having a self-loop and an arc to $s_{uv}$.
$\Delta$ is a function of $n_0$ to be decided later.
Assign weight $d-1$ to each self-loop, and weight $1$ to any other arc.
Let $n$ be the number of nodes in $G$; clearly $n = \Theta(n_0 \Delta)$.
$G$ is now a version of $G_0$ where the random walk slows down by a factor roughly $d\Delta = \Theta(\Delta)$, since moving between two nodes $u$ and $u'$ that were neighbors in $G_0$ now takes $\Theta(\Delta)$ steps in expectation.
The mixing time of $G$ is therefore $\tau = \Theta(\tau_0 \Delta)$.
This holds also for an arbitrary algorithm: any two neighbors of a star center are indistinguishable until they are visited, since their transition probabilities are the same ($1/\Delta$ from the center, and $1/d$ to the center).
The same holds for the neighbors of the original nodes of $G_0$, whose transition probabilities from/to such a node are $1/d$ and $1/\Delta$ respectively.
Therefore, starting from any node in $G$ any algorithm needs $\Omega(\sqrt{\Delta n})$ queries to estimate $n$.
Finally, if $\pi$ is the stationary distribution of the random walk, then $\|\pi\|= \Theta(\sqrt{\Delta/n})$, since there are $\Delta/n$ nodes (the centers of the stars) all having the same mass which is also asymptotically larger than the mass of any other node, and which in aggregate is $\Omega(1)$.

Now consider the graph $G'=(\Delta,\frac{n_0}{2})$, which has half the nodes of $G$.
Choose any node $u \in G'$.
We now add $k$ nodes to $G'$, with $k = \Theta(n)$, so that it has exactly the same number of nodes $n$ as $G$.
Finally, we add $k$ arcs between each of those nodes and $u$, each one of weight $\frac{\epsilon}{k}$ for some $\epsilon > 0$.
Both the overall mass of these $k$ nodes, and the probability of walking to any of them from $u$, is then less than $\epsilon$.
Therefore, the mixing time of $G'$ is essentially unaltered, as well as its stationary distribution.
However, any node in $G$ has roughly half the mass of its ``homologue'' in $G'$.
Therefore, to estimate the mass of any given node in $G$, one must distinguish between $G$ and $G'$, i.e.\ determine whether the graph at hand comes from $G(\Delta,n_0)$ or $G(\Delta,\frac{n_0}{2})$; and as we have seen this requires $\Omega(\sqrt{\Delta n})$ queries.

Now to the bounds. Note that
$\tau\|\pi\|^{-1} = \Theta\big(\tau_0 \Delta \sqrt{n/\Delta}\big) = O(\sqrt{\Delta n} \ln{n})$,
and
$\sqrt{\tau n} = \Theta\big(\sqrt{\Delta \tau_0 n}\big) = O\big(\sqrt{\Delta n}\sqrt{\ln{n}}\big)$;
thus $\sqrt{\Delta n}$ is in both $\Omega(\tau\|\pi\|^{-1} / \ln{n})$ and $\Omega(\sqrt{\tau n / \ln{n}})$.
Since $\tau =\Theta(\Delta\tau_0) = \Theta(\frac{n}{n_0}\ln{n_0})$ and $\Delta = \Theta(\frac{n}{n_0})$, by appropriately choosing $n_0 \in \Theta(1) \cap \Theta(n)$ we can make $\tau$ range from $\Theta(\ln{n})$ to $\Theta(n)$.
In the same way, since $\|\pi\| = \Theta(\sqrt{\Delta/n}) = \Theta(\sqrt{1/n_0})$ we can make $\|\pi\|$ range from $\Theta(\frac{1}{\sqrt{n}})$ to $\Theta(1)$ (although not independently of $\tau$).

\section{Conclusions}
\label{sec:conc}
We have given improved, optimal algorithms for approximating the stationary probability of a given state in a time-reversible Markov chain, and for approximating the sum of nonnegative real vectors by weighted sampling.
Although time-reversible chains are of clear relevance, extending our results to other classes of Markov chains is an intriguing open question.
We have also shown that the footprint of our algorithms in terms of number of distinct states visited is tied to the concentration of the number of distinct states visited by the chain; investigating such a concentration is thus an obvious line of future research.

\clearpage

\bibliographystyle{plain}
\bibliography{biblio-arxiv}

\clearpage

\input{apx}



\end{document}

%% file: apx.tex
\section{Appendix}
\label{sec:apx}

\subsection{Probability bounds}
\label{apx:chernoff_bounds}
This appendix provides Chernoff-type probability bounds that are repeatedly used in our analysis; these bounds can be found in e.g.~\cite{Auger&2011}, and can be derived from~\cite{Panconesi&1997}.

Let $X_1,\ldots,X_n$ be binary random variables. We say that $X_1,\ldots,X_n$ are non-positively correlated if for all $I \subseteq \{1,\ldots,n\}$ we have:
\begin{align}
Pr[\forall i \in I: X_i=0] &\leq \prod_{i \in I} Pr[X_i=0] \\
Pr[\forall i \in I: X_i=1] &\leq \prod_{i \in I} Pr[X_i=1]
\end{align}
The following lemma holds:
\begin{lemma}
\label{lem:chernoff}
Let $X_1,\ldots,X_n$ be independent or, more generally, non-positively correlated binary random variables. Let $a_1,\ldots,a_n \in [0,1]$ and $X=\sum_{i=1}^{n}a_i X_i$. Then, for any $\epsilon > 0$, we have:
\begin{align}
Pr[X < (1-\epsilon)\E[X]] &< e^{-\frac{\epsilon^2}{2}\E[X]} \\
Pr[X > (1+\epsilon)\E[X]] &< e^{-\frac{\epsilon^2}{2+\epsilon}\E[X]} 
\end{align}
\end{lemma}
Note that Lemma~\ref{lem:chernoff} applies if $X_1,\ldots,X_n$ are indicator variables of mutually disjoint events, or if they can be partitioned into independent families $\{X_1,\ldots,X_{i_1}\}$, $\{X_{i_1+1},\ldots,X_{i_2}\}$, \ldots of such variables. 

\subsection{A lower bound for non-time-reversible Markov chains}
\label{sub:general_impossible}
\begin{lemma}
For any functions $\tau(n) = \omega(1)$ and $p(n) = o(\frac{1}{n})$ there exists a family of ergodic non-time-reversible Markov chains on $n$ states having mixing time $\tau = \Theta(\tau(n))$, and containing a state $v$ with $\pi(v) = \Theta(p(n))$ such that any algorithm needs $\Omega(\frac{\tau}{\pi(v)})$ calls to \step() to estimate $\pi(v)$ within constant multiplicative factors with constant probability.
\end{lemma}
\begin{proof}
Consider a chain with state space $\{u\} \cup \{u_1,\ldots,u_{n-1}\}$ and
the following transition probabilities (we assume $n$ large enough to set in $[0,1]$ any quantity where needed).
For $u$, set $p_{uu} = 1-\frac{(n-1)p(n)}{\tau(n)}$, and $p_{uu_i} = \frac{p(n)}{\tau(n)}$ for all $i=1,\ldots,n-1$.
For all $i=1,\ldots,n-1$, set $p_{u_i u_i} = 1-\frac{1}{\tau(n)}$ and $p_{u_i u} = \frac{1}{\tau(n)}$.
The chain is clearly ergodic.
Note that $\frac{(n-1)p(n)}{\tau(n)} = o(\frac{1}{\tau(n)})$ and therefore the expected time to leave $u$ is asymptotically larger than the expected time to leave any of the $u_i$.
One can then check that (i) $\pi(u_i) = \Theta(p(n))$, and (ii) the mixing time is $\tau = \Theta(\tau(n))$ (essentially, the expected time to leave the $u_i$).
Pick any $u_i$ as target state $v$.
Suppose now to alter the chain as follows: pick some $u_j \ne v$ and set $p_{u_j v} = 1$.
The new stationary probability of $v$ would then be roughly $2\pi(v)$.
However one cannot distinguish between the two chains with constant probability with less than
$\Omega(\frac{\tau}{\pi(v)})$ \step() calls.
Indeed, to distinguish between them one must at least visit $u_j$ (and then perform e.g.\ \probe($u_j,v$)).
Since $u$ is the only state leading to $u_j$ with positive probability, one must invoke \step($u$) until it returns $u_j$.
But $p_{uu_j}=\frac{p(n)}{\tau(n)}$, hence one needs $\Omega(\frac{\tau(n)}{p(n)}) = \Omega(\frac{\tau}{\pi(v)})$ calls in expectation.
The construction can be adapted to any constant approximation factor by adding more transitions towards $v$.
\end{proof}

\clearpage
\subsection{Pseudocode of \taupiest}
\label{apx:taupiest}
\renewcommand{\thealgorithm}{}
\begin{algorithm}[h!]
\label{alg:taupiest}
\small
\caption{\taupiest($v, \epsilon, \delta$)}
\begin{algorithmic}[1]
\State $S \leftarrow \emptyset$ \Comment{distinct states visited so far}
\State $w_S\leftarrow 0$ \Comment{$\sum_{u \in S} \gamma_u$ for the current $S$}
\State $w\leftarrow 0$ \Comment{will accumulate $w_S$}
\State $r\leftarrow 0$ \Comment{number of repeats witnessed}
\State $\kerr \leftarrow 4\lceil\frac{2+4.4\epsilon}{\epsilon^2}\ln{\!\frac{3}{\delta}}\rceil$ \Comment{halting threshold on the number of repeats}
\vspace*{0.3em}
\While{$r < \kerr$}
\State $w\leftarrow w + w_S$
\State $(u, \gamma_u) \leftarrow $ sample drawn by walking $t$ steps starting from $v$
\If {$u \in S$} \Comment{detect repeat}
\State $r \leftarrow r+1$ 
\Else
\State $S \leftarrow S \cup \{u\}$
\State $w_S \leftarrow w_S + \gamma_u$
\EndIf
\EndWhile
\State \textbf{return} $r/w$ \Comment{estimate of $1/\gamma$, i.e.\ of $\pi(v)$}
\end{algorithmic}
\end{algorithm}

\subsection{Pseudocode of \taunest}
\label{apx:fma}
\renewcommand{\thealgorithm}{}
\begin{algorithm}[h]
\small
\caption{\taunest($\epsilon, \delta, v$)}
\begin{algorithmic}[1]
\State $S \leftarrow \{v\}$ \Comment{distinct states visited so far}
\State $D \leftarrow \{v:1\}$ \Comment{dictionary mapping $u$ to $\gamma_u$}
\State $w_S\leftarrow 1$ \Comment{$\sum_{u \in S} \gamma_u$ for the current $S$}
\State $w\leftarrow 0$ \Comment{will accumulate $w_S$}
\State $r\leftarrow 0$ \Comment{number of repeats witnessed}
\State $\kerr \leftarrow 4\lceil\frac{2+4.4\epsilon}{\epsilon^2}\ln{\!\frac{3}{\delta}}\rceil$ \Comment{halting threshold on the number of repeats}
\State $u \leftarrow v$ \Comment{current walk state}
\While{$r < \kerr$}
\State $w\leftarrow w + w_S$
\State $N \leftarrow \emptyset$ \Comment{new states visited}
\For{$i = 1$ to $t$} 
\State $\bar{u} \leftarrow$ \step($u$)
  \If{$\bar{u} \notin D$} \Comment{$\bar{u}$ never visited before}
  \State $D[\bar{u}] = D[u]\,\cdot\,$\probe($u,\bar{u}$)$ / $\probe($\bar{u},u$)
  \State $N \leftarrow N \cup \{\bar{u}\}$
  \EndIf
\State $u \leftarrow \bar{u}$
\EndFor
\If {$u \in S$} \Comment{detect repeat}
\State $r \leftarrow r+1$ 
\EndIf
\State $w_S \leftarrow w_S + \sum_{u \in N} D[u]$
\State $S \leftarrow S \cup N$ 
\EndWhile
\State \textbf{return} $r/w$ \Comment{estimate of $1/\gamma$, i.e.\ of $\pi(v)$}
\end{algorithmic}
\end{algorithm}

\clearpage
\subsection{Experiments}
\label{sub:exp}
We experimentally evaluate \taupiest\ and \taunest\ against the algorithms of Lee et al.~\cite{Lee&2013} and Banerjee et al.~\cite{Lofgren&2015b} (see Section~\ref{sub:rel}).
All algorithms were ran on synthetic time-reversible Markov chains on $1$M states, created as follows.
We start from an undirected torus graph (i.e.\ a grid with periodic boundary) of $n = 1000 \times 1000$ nodes.
We then add $0.01n$ edges between random pairs of nodes, to reduce the mixing time and thus the cost (and running time) of the algorithms.
We add self-loops to all nodes to ensure ergodicity.
Finally, we weight the arcs according to two distributions:
\begin{itemize}
\item $\pi_{U}$ (uniform): each arc has weight $1$.
The norm is $\|\pi_{U}\| = 0.001$, or essentially $1/\sqrt{n}$.
\item $\pi_{S}$ (skewed): each arc is given an independent weight $1/X$ where $X \sim \mathcal{U}(0,1]$.
The norm is $\|\pi_{S}\| \simeq 0.07$.
\end{itemize}
For each weighted graph, we consider the time-reversible chain of the associated random walk.

We picked $v=0$ as the target node, which is equivalent to any other one (and indeed repeating the experiments on other nodes yielded the same results).
For all algorithms we set $\delta=0.1$.
For the algorithm of Banerjee et al.\ we set the minimum detection threshold at $\epsilon \pi(v)$, and for all other algorithms we set $\epsilon=0.25$.
One must then fix the random walk length: $t$ in our algorithms, $\ell$ in Banerjee et al., and $1/\Delta$ in Lee et al.
Setting the lengths to $\simeq \tau \ln(n)$ would make all algorithms satisfy the desired guarantees.
Since we do not know $\tau$, for each algorithm we proceed as follows.
We initially set the length of random walks to $l=10$.
We then perform three independent executions of the algorithm.
If all three executions return an estimate $\hat{\pi}(v)$ within a multiplicative factor $(1 \pm \epsilon)$ of $\pi(v)$, then we stop.
Otherwise, we increase $l$ by a factor $\sqrt{2}$ and repeat.
For each value of $l$ we record the average relative error $\hat{\epsilon} = \frac{|\hat{\pi}(v) - \pi(v)|}{\pi(v)}$ and the average total number of \step() and \probe() calls.
Figure~\ref{fig:1} shows how $\hat{\epsilon}$ decreases as the number of calls increases.

\begin{figure}[h]
\begin{minipage}{0.49\textwidth}
\centering
\includegraphics[width=\textwidth]{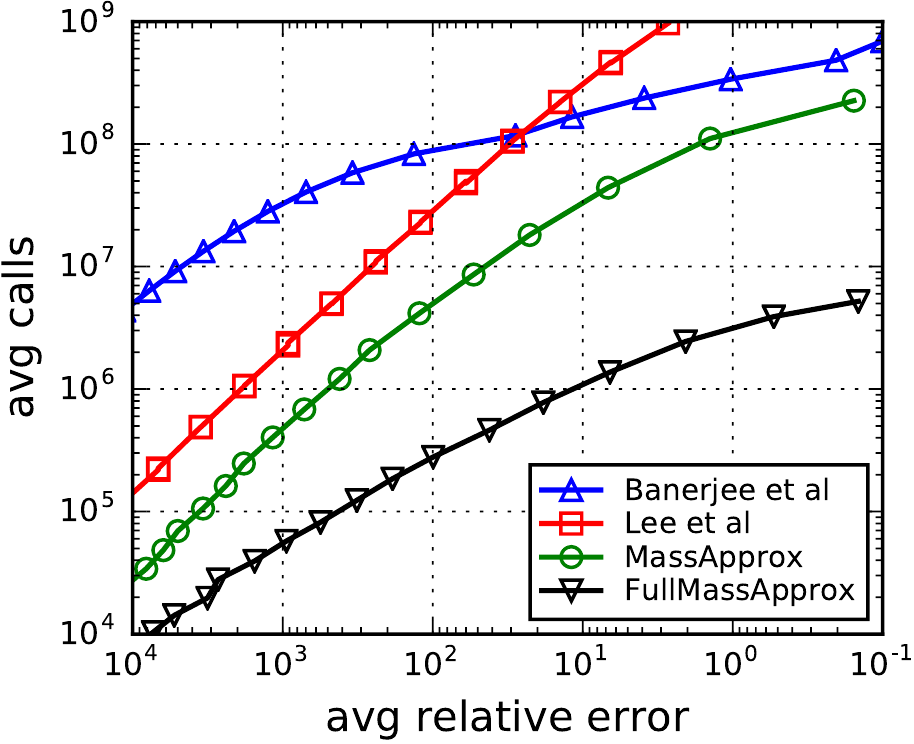}
\end{minipage}
\begin{minipage}{0.02\textwidth}
\hspace{0.01\textwidth}
\end{minipage}
\begin{minipage}{0.49\textwidth}
\centering
\includegraphics[width=\textwidth]{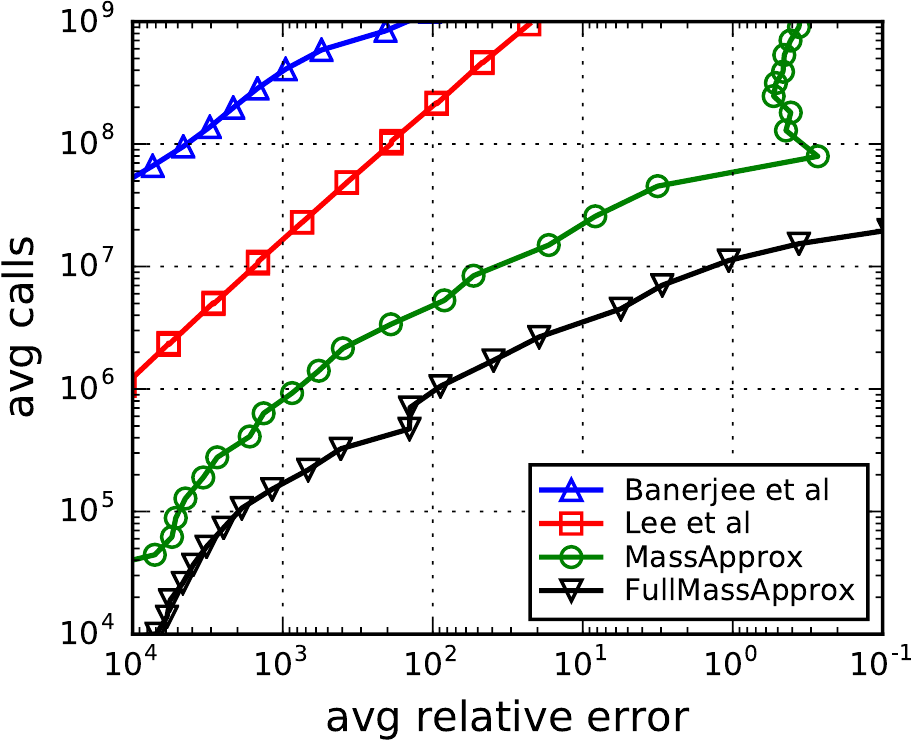}
\end{minipage}
\caption{cost incurred by the algorithms as their estimates converge towards $\pi(v)$. Left: chain with uniform distribution $\pi_{U}$. Right: chain with skewed distribution $\pi_S$.}
\label{fig:1}
\end{figure}

\taupiest\ and \taunest\ are the fastest candidates in all cases.
In the uniform chain, \taupiest\ is approached by the algorithm of Banerjee et al.\ at high accuracies.
This seems a confirmation of theory: \taupiest\ has complexity $\tilde{O}(\tau n^{-0.5})$ on a chain with uniform distribution, and the algorithm of Banerjee et al.\ has complexity $\tilde{O}(\tau^{1.5} n^{-0.5})$ on the ``typical'' target state with mass $\pi(v) \approx 1/n$.
If $\tau$ is not exceedingly large, the two complexities can translate into close performance in practice.
On the other hand, \taunest\ is neatly more efficient than previous algorithms.
To obtain a fairly accurate estimate of $\pi(v)$, say $\pm 50 \%$, it improves on their performance by two orders of magnitude -- and possibly by more on the skewed chain.
These results suggest that our algorithms are not only of theoretical interest, but also of practical value.
A final observation is that \taunest\ outperforms also \taupiest\ on the uniform chain. The complexity bounds we have are the same for both algorithms, but perhaps \taunest\ takes advantage of some specific structural properties of the chain we have used, which makes its complexity drop further.

%% file: paper-arxiv.bbl
\begin{thebibliography}{10}

\bibitem{Alon&2010}
Noga Alon, Ori Gurel-Gurevich, and Eyal Lubetzky.
\newblock Choice-memory tradeoff in allocations.
\newblock {\em The Annals of Applied Probability}, 20(4):1470--1511, 2010.

\bibitem{Auger&2011}
Anne Auger and Benjamin Doerr, editors.
\newblock {\em Theory of Randomized Search Heuristics: Foundations and Recent
  Developments}, volume~1.
\newblock World Scientific Publishing Co., Inc., 2011.

\bibitem{Lofgren&2015b}
Siddhartha Banerjee and Peter Lofgren.
\newblock Fast bidirectional probability estimation in {M}arkov models.
\newblock In {\em Proc.\ of NIPS}, pages 1423--1431, 2015.

\bibitem{Bonacich&2001}
Phillip Bonacich and Paulette Lloyd.
\newblock Eigenvector-like measures of centrality for asymmetric relations.
\newblock {\em Social Networks}, 23(3):191 -- 201, 2001.

\bibitem{Borgs&2012}
Christian Borgs, Michael Brautbar, Jennifer~T. Chayes, and Shang-Hua Teng.
\newblock A sublinear time algorithm for {P}age{R}ank computations.
\newblock In {\em Proc.\ of WAW}, pages 41--53. Springer, 2012.

\bibitem{Borgs&2014}
Christian Borgs, Michael Brautbar, Jennifer~T. Chayes, and Shang{-}Hua Teng.
\newblock Multiscale matrix sampling and sublinear-time {P}age{R}ank
  computation.
\newblock {\em Internet Mathematics}, 10(1-2):20--48, 2014.

\bibitem{Bressan&2018}
Marco Bressan, Enoch Peserico, and Luca Pretto.
\newblock On approximating the stationary distribution of time-reversible
  {M}arkov chains.
\newblock In {\em Proc.\ of STACS}, 2018.

\bibitem{Chung&2012}
Kai-Min Chung, Henry Lam, Zhenming Liu, and Michael Mitzenmacher.
\newblock Chernoff-{H}oeffding bounds for {M}arkov chains: Generalized and
  simplified.
\newblock In {\em Proc.\ of STACS}, pages 124--135, 2012.

\bibitem{Freedman1975}
David~A. Freedman.
\newblock On tail probabilities for martingales.
\newblock {\em The Annals of Probability}, 3(1):100--118, 1975.

\bibitem{golub2012matrix}
Gene~H. Golub and Charles F.~Van Loan.
\newblock {\em Matrix Computations}.
\newblock Matrix Computations. Johns Hopkins University Press, 2012.

\bibitem{Hastings1970}
Wilfred~K. Hastings.
\newblock Monte {C}arlo sampling methods using {M}arkov chains and their
  applications.
\newblock {\em Biometrika}, 57(1):97--109, 1970.

\bibitem{Lee&2013}
Christina~E. Lee, Asuman Ozdaglar, and Devavrat Shah.
\newblock Computing the stationary distribution, locally.
\newblock In {\em Proc.\ of NIPS}, pages 1376--1384, 2013.

\bibitem{Lee&2014}
Christina~E. Lee, Asuman~E. Ozdaglar, and Devavrat Shah.
\newblock Solving systems of linear equations: Locally and asynchronously.
\newblock {\em CoRR}, abs/1411.2647, 2014.

\bibitem{Levin&2006}
David~A. Levin, Yuval Peres, and Elizabeth~L. Wilmer.
\newblock {\em Markov chains and mixing times}.
\newblock American Mathematical Society, 2006.

\bibitem{Lofgren&2015}
Peter Lofgren, Siddhartha Banerjee, and Ashish Goel.
\newblock Bidirectional {PageRank} estimation: from average-case to worst-case.
\newblock In {\em Proc.\ of WAW}, pages 164--176, 2015.

\bibitem{Lofgren&2014b}
Peter~A. Lofgren, Siddhartha Banerjee, Ashish Goel, and C.~Seshadhri.
\newblock {FAST-PPR}: Scaling personalized {P}age{R}ank estimation for large
  graphs.
\newblock In {\em Proc.\ of ACM KDD}, pages 1436--1445, 2014.

\bibitem{Motwani&2007}
Rajeev Motwani, Rina Panigrahy, and Ying Xu.
\newblock Estimating sum by weighted sampling.
\newblock In {\em Proc.\ of ICALP}, pages 53--64, 2007.

\bibitem{Panconesi&1997}
Alessandro Panconesi and Aravind Srinivasan.
\newblock Randomized distributed edge coloring via an extension of the
  {C}hernoff--{H}oeffding bounds.
\newblock {\em SIAM Journal on Computing}, 26(2):350--368, 1997.

\bibitem{Rubinfeld&2011}
Ronitt Rubinfeld and Asaf Shapira.
\newblock Sublinear time algorithms.
\newblock {\em SIAM Journal on Discrete Mathematics}, 25(4):1562--1588, 2011.

\bibitem{Shyamkumar&}
Nitin Shyamkumar, Siddhartha Banerjee, and Peter Lofgren.
\newblock Sublinear estimation of a single element in sparse linear systems.
\newblock In {\em 2016 Annual Allerton Conference on Communication, Control,
  and Computing (Allerton)}, pages 856--860, 2016.

\end{thebibliography}
